\documentclass[sigconf]{acmart}

\renewcommand\footnotetextcopyrightpermission[1]{} 
\usepackage{amsmath}
\usepackage{bm}
\usepackage{nicefrac}
\usepackage{booktabs}
\usepackage{array}
\usepackage{multirow}
\usepackage{threeparttable}
\usepackage{makecell}
\usepackage[procnumbered,ruled,vlined,linesnumbered]{algorithm2e}
\usepackage{siunitx}
\usepackage{stfloats}
\usepackage{graphicx}
\usepackage{subfigure}
\usepackage{hyperref}
\usepackage{array}
\usepackage{epstopdf}
\usepackage{mweights}
\usepackage{balance}

\newtheorem{problem}{Problem}

\newcommand{\PreserveBackslash}[1]{\let\temp=\\#1\let\\=\temp}
\newcolumntype{C}[1]{>{\PreserveBackslash\centering}p{#1}}
\newcolumntype{R}[1]{>{\PreserveBackslash\raggedleft}p{#1}}
\newcolumntype{L}[1]{>{\PreserveBackslash\raggedright}p{#1}}

\def\defeq{\stackrel{\mathrm{def}}{=}}

\def\sizeof#1{\left|#1  \right|}

\def\eps{\epsilon}

\def\trace#1{\mathrm{Tr} \left(#1 \right)}
\def\norm#1{\left\| #1 \right\|}
\def\smallnorm#1{\| #1 \|}

\def\calG{\mathcal{G}}

\def\calN{\mathcal{N}}

\def\norm#1{\left\| #1 \right\|}

\def\kh#1{\left( #1 \right)}

\def\ceil#1{\left\lceil #1 \right\rceil}

\def\defeq{\stackrel{\mathrm{def}}{=}}

\newcommand{\removelatexerror}{\let\@latex@error\@gobble}

\newcommand{\rea}{\mathbb{R}}

\newcommand\LL{\bm{\mathit{L}}}
\newcommand\Otil{\widetilde{O}}

\def\defeq{\stackrel{\mathrm{def}}{=}}
\def\trace#1{\mathrm{Tr} \left(#1 \right)}
\def\sizeof#1{\left|#1  \right|}

\newcommand{\one}{\mathbf{1}}

\newcommand\WW{\boldsymbol{\mathit{W}}}
\newcommand\XX{\boldsymbol{\mathit{X}}}
\newcommand\yy{\boldsymbol{\mathit{y}}}

\newcommand\xx{\boldsymbol{\mathit{x}}}
\newcommand\aaa{\boldsymbol{\mathit{a}}}

\newcommand\bb{\boldsymbol{\mathit{b}}}
\newcommand\cc{\boldsymbol{\mathit{c}}}

\newcommand\ee{\boldsymbol{\mathit{e}}}
\newcommand\pp{\boldsymbol{\mathit{p}}}

\newcommand\hh{\boldsymbol{\mathit{h}}}

\renewcommand\SS{\boldsymbol{\mathit{S}}}
\renewcommand\AA{\boldsymbol{\mathit{A}}}
\newcommand\BB{\boldsymbol{\mathit{B}}}
\newcommand\bBB{\bar{\BB}}
\newcommand\CC{\boldsymbol{\mathit{C}}}

\newcommand\DD{\boldsymbol{\mathit{D}}}

\newcommand\EE{\boldsymbol{\mathit{E}}}
\newcommand\PP{\boldsymbol{\mathit{P}}}

\newcommand\YY{\boldsymbol{\mathit{Y}}}

\newcommand\QQ{\boldsymbol{\mathit{Q}}}

\newcommand\vvv{\boldsymbol{\mathit{v}}}
\newcommand{\SDDMSolver}{\textsc{Solve}}

\DeclareMathOperator*{\argmax}{arg\,max}

\DontPrintSemicolon
\SetKw{KwAnd}{and}
\SetFuncSty{textsc}
\SetKwInOut{Input}{Input\ \ \ \ }

\SetKwInOut{Output}{Output}

\usepackage{tabularx}
\usepackage{stfloats}
\usepackage{booktabs}
\begin{document}
	\fancyhead{}
	\title{Maximizing  Influence of  Leaders in Social Networks}
\author{Xiaotian Zhou and Zhongzhi Zhang\footnotemark}
\affiliation{
\institution{Shanghai Key Laboratory of Intelligent Information Processing, Fudan University, Shanghai 200433, China
\\School of Computer Science, Fudan University, Shanghai 200433, China}
		\city{}
		\country{}
}
\email{{20210240043,zhangzz}@fudan.edu.cn}

\begin{abstract}

The operation of adding edges has been frequently used to the study of opinion dynamics in social networks for various purposes. In this paper, we consider the  edge addition problem for the DeGroot model of opinion dynamics in a social network with $n$ nodes and $m$ edges, in the presence of a small number $s \ll n$ of competing leaders with binary opposing opinions 0 or 1. Concretely, we pose and investigate the problem of maximizing the equilibrium overall opinion by creating $k$ new edges in a candidate edge set, where each edge is incident to a 1-valued leader and a follower node. We show that  the objective function is monotone and submodular. We then propose a simple greedy algorithm with an approximation factor $(1-\frac{1}{e})$ that approximately solves the problem in $O(n^3)$ time. Moreover, we provide a fast algorithm with a $(1-\frac{1}{e}-\epsilon)$ approximation ratio and $\tilde{O}(mk\epsilon^{-2})$ time complexity for any $\epsilon>0$, where $\tilde{O}(\cdot)$ notation suppresses the ${\rm poly} (\log n)$ factors. Extensive experiments demonstrate that our second approximate algorithm is efficient and effective, which  scales to large networks with more than a million nodes.
\end{abstract}

%
%
	
	
	

\maketitle
\renewcommand{\thefootnote}{*}
\footnotetext[1]{Corresponding author. Zhongzhi Zhang is also with Shanghai Blockchain Engineering Research Center, as well as Research Institute of Intelligent Complex Systems, Fudan University, Shanghai 200433.}

\section{Introduction}
It has been heavily studied in the community of social sciences how opinions evolve and shape through social interactions among individuals with potentially different opinions~\cite{De74,FrJo90}. In the current digital age, social media and online social networks provide unprecedented access to social interactions and individuals' opinions, leading to a fundamental change of the way people propagate, share, and shape opinions~\cite{Le20}. As such, in the past years the problem of opinion dynamics has   received considerable attention from scientists in various disciplines~\cite{JiMiFrBu15,DoZhKoDiLi18,AnYe19}, including the AI community~\cite{DaGoMu14,FoPaSk16,AuFeGr18}. In order to understand the evolution and shaping of opinions, numerous models have been introduced~\cite{NoViTaHu20}, among which the DeGroot model~\cite{De74} is a popular one. In the  DeGroot model, the final opinions of  individuals reach consensus. Since its establishment, the DeGroot model has been extended or generalized by incorporating different factors affecting opinion formation or agreement~\cite{JiMiFrBu15,DoZhKoDiLi18}.

Within the area of opinion dynamics, in addition to the development of models explaining the spread and formation of opinions, the problem of effectively shifting opinions in a social network has also become an active direction of study in recent years~\cite{GiTeTs13,MeAsDaAmAn13,VaFaFr14,AbKlPaTs18,XuHuWu20}, 
since it is closely related to various practical application settings, such as public health campaigns, political candidates, and product marketing. Most of previous work focus on the operations on individuals in a social network so as to optimize different opinions, for example, to maximize or minimize the overall opinion of the network. However, link operations for optimizing the overall opinion has not been well studied and understood. In this paper we consider the problem of maximizing the overall opinion by adding edges, namely, making friends in a social network.


We consider a variant of the DeGroot model~\cite{MeAsDaAmAn13,VaFaFr14,MaPa19} for opinion dynamics on a social network with $n$ nodes and $m$ edges, where nodes represent individuals, and edges denote interactions among individuals. In this extended DeGroot model, the $n$ individuals are classified two classes $S$ and $F$: $S$ includes a small number $s\ll n$ of leader nodes, while the remaining $n-s$ nodes in $F$ are followers. In a real social network, a leader could be a paid promoter of a certain product or political position. The leader set $S$ is further categorized into to two subsets $S_0$ and $S_1$, which represent firm supporters of the opposing parties, containing nodes with opposing opinions 0 and 1, respectively. In the model, each $i$ node at time $t$ has a nonnegative scalar-value opinion $\xx_i(t)$ in the interval $[0,1]$. When $i$ is a follower, $\xx_i(t)$ evolves as a weighted average of the opinions $\xx_j(t-1)$ of $i$'s neighbors $j$. When $i$ is a leader, $\xx_i(t)$ never changes over time. Finally, the opinion of each follower converges to a value between 0 and 1.

We address the following problem based on the aforementioned DeGroot model: Given a social network and a positive integer $k$,   
how to create $k$ edges, each connecting a 1-valued leader and a follower, so that the overall stead-state opinion is maximized. This problem is at variance with existing work in terms of both formulation and solution. The main contributions of our work are as follows. First, we show that 
the objective function is monotone and submodular, which are established using a technique completely different from existing ones. Then, we propose two approximation algorithms solving the combinatorial optimization problem based on a greedy strategy of adding edges. The two greedy algorithms are guaranteed to have, respectively, approximation ratios $\left(1-\frac{1}{e}\right)$ and $\left(1-\frac{1}{e}-\eps\right)$, where $ \eps>0$ is the error parameter. The first algorithm has time complexity $O(n^3)$, while the second algorithm has complexity $\Otil (mk\eps^{-2})$, where the notation $\Otil (\cdot)$ hides ${\rm poly} (\log n)$ factors. Finally, comprehensive experiments on various real networks are performed to demonstrate the efficiency and effectiveness of our algorithms, as well as their performance advantages, compared with several baseline strategies of edge addition.

\section{Related Work}

In this section, we  briefly review the literature  related to our work.

The model under consideration is based on a variant of the popular DeGroot model for opinion dynamics~\cite{De74}.  Since
the time the  DeGroot model was  established, numerous extensions or  variants have been proposed by considering different aspects or processes affecting the limiting opinions.  For example, the  Friedkin-Johnson (FJ) model~\cite{FrJo90}  is a generalization of the DeGroot model, where each individual has two opinions, internal opinion and expressed opinion.  
After a long time evolution, the  expressed opinions of   individuals  in the FJ model converge but often do not reach agreement. Another  extension of the DeGroot model is the Altafini model~\cite{Al13},  where it is supposed that
the interactions between individuals are not always cooperative, but sometimes antagonistic.

Many existing studies also pay attention to variants of the  DeGroot model  by selecting  1-leaders in the presence of  competing 0-leaders, which are formulated for optimizing different objectives such as minimizing disagreement and polarization~\cite{YiPa20}, maximizing the diversity~\cite{MaPa19} and the total opinion~\cite{MeAsDaAmAn13,VaFaFr14,MaAb19}. Similar  opinion maximization problem was also studied  for the FJ model by using different strategies, including identifying a given number of 1-leaders~\cite{GiTeTs13}, as well as modifying individual's internal opinions~\cite{XuHuWu20} or susceptibility to persuasion~\cite{AbKlPaTs18,ChLiSo19}. Although the problem we address is also the maximization of the overall opinion, our strategy is optimally  selecting edges to add,  instead of leader selection.

Admittedly,  as a practical approach of graph edit,  edge addition operation has been extensively used for different   application purposes, such as improving the centrality of a node~\cite{CrDaSeVe16,ShYiZh18,DaOlSe19} and maximizing the number of spanning trees~\cite{LiPaYiZh20}. For a social network,  creating  edges corresponds to  making friends.
In the field of opinion dynamics, the problem of  adding  edges has  also been  studied in order to optimize different objectives.
For example, the edge addition strategy was exploited in~\cite{BiKlOr11,BiKlOr15}, aiming at minimizing the social cost at equilibrium in the FJ model. Again for instance, in~\cite{GaDeGiMa17} and~\cite{ChLiDe18},  creating edges was adopted  to reduce, respectively, controversy and risk of conflict. Finally,  edge recommendation was used in~\cite{AmSi19} with an aim to strategically fight opinion control in a social network. Motivated in part by these work, we exploit the manner of adding edges to maximize the overall opinion. Departing  from the literature in the area of opinion dynamics, we  present a nearly linear  algorithm for evaluating the overall opinion, which is proved to have a guaranteed approximation ratio.

\section{Preliminary}

This section is devoted to a brief introduction to some useful notations and tools, in order to facilitate the description of  problem formulation and  algorithms.

\subsection{Notations}

We denote scalars in $\rea$ by normal lowercase letters like $a,b,c$, sets by normal uppercase letters like $A,B,C$, vectors by bold lowercase letters like $\aaa, \bb, \cc$, and matrices by bold uppercase letters like $\AA, \BB, \CC$. We use $\one$ to denote  the vector  of appropriate dimensions with all entries being ones, and use $\ee_i$  to denote the $i^{\rm th}$ standard basis vector of appropriate dimension. Let $\aaa^\top$ and $\AA^\top$  denote, respectively, transpose of  vector $\aaa$ and matrix  $\AA$. Let $\trace \AA$ denote the trace of matrix $\AA$. We use  $\AA_{[i, :]}$ and $\AA_{[:, j]}$ to denote,  respectively, the $i^{\rm th}$ row  and the $j^{\rm th}$ column of matrix $\AA$. We write $\AA_{i,j}$ to denote the entry at row $i$ and column $j$ of $\AA$ and we use $\aaa_i$ to denote the $i^{\rm th}$ element of vector $\aaa$. We write sets in matrix subscripts to denote submatrices. For example, $\AA_{F,S}$ represents the submatrix of $\AA$ with row and  column indices in $F$ and  $S$, respectively. In addition,
we use $ \AA_{F}$ denotes the submatrix of $\AA$ obtained from $\AA$ with both the row and column indices in $F$.  

For a matrix $\XX\in\mathbb{R}^{m\times n}$ with entries $\XX_{i,j}$, $i=1,2,\cdots,m$ and $j=1,2,\cdots,n$, its Frobenius norm $\norm{\XX}_{F}$ is
\begin{equation*}
	\norm{\XX}_{F}\defeq\sqrt{\sum_{i=1}^{m}\sum_{j=1}^{n}\XX_{i,j}^2}=\sqrt{\trace{\XX^\top\XX}}.
\end{equation*}


\begin{definition}
Let $a,b>0$ be positive scalars.  $a$ is called an $\epsilon$-approximation ($0 < \epsilon < 1/2$) of $b$ if the following relation holds: $  (1-\epsilon)b \leq a \leq (1+\epsilon)b$, which is denoted $a \mathop{\approx}\limits^{\epsilon} b$ for simplicity.
\end{definition}
For a finite set $X$, let $2^X$ be the set of all  subsets of $X$.  Let $f: 2^X \to \mathbb{R}$ be a set function on $2^X$.  For any subset $T \subseteq W \subseteq X$ and any element $a \in X \setminus W$, the function $f$ is called \textit{submodular} if
$f(T \cup \{a\}) -f(T) \geq f(W \cup \{a\}) -f(W)$. $f: 2^X \to \mathbb{R}$ is called \textit{monotone increasing} if for any subset $T \subseteq W \subseteq X$, 	$f(T) \leq f(W)$ holds true.

\subsection{Graphs and Related Matrices}


Let $\calG = (V,E)$ be a connected undirected network with $n = |V|$ nodes and  $m = |E|$ edges, where $V$ is the set of nodes and $E \subseteq V \times V$ is the set of edges. The adjacency relation of all nodes in $\calG$ is encoded in the adjacency matrix $\AA$,  whose entry $\AA_{i, j}=1$ if  $i$ and $i$ are adjacent, and $\AA_{i, j}=0$ otherwise. Let $\calN_i$ be the set of neighbouring nodes of  $i$. Then degree $d_i$ of node  $i$ is  $\sum_{j=1}^n \AA_{i, j}=\sum_{j\in \calN_i} \AA_{i, j}$. The degree diagonal matrix of   $\calG$ is $\DD={\rm diag}(d_1, d_2, \cdots, d_n)$, and the Laplacian matrix $\LL$ of  $\calG$ is $\LL = \DD - \AA$. If we fix an arbitrary orientation for all  edges in $\calG$, then we can define the signed edge-node incidence matrix  $\BB_{m\times n}$ of  graph $\calG$,  whose entries are defined as follows:  $\BB_{e,u}= 1$ if node $u$ is the  head of edge $e$, $\BB_{e,u} = -1$ if $u$ is the tail of $e$, and $\BB_{e,u} = 0$ otherwise.  For an oriented edge $e\in E$ with end nodes   $u$ and $v$, we define $\bb_e = \bb_{u,v}=\ee_u - \ee_v$ if $u$ and $v$ are, respectively, the head and tail of $e$. Then $\LL$ can be written as $\LL = \BB^\top \BB$ or $\LL = \sum\nolimits_{e\in E} \bb_e \bb_e^\top$,  meaning that $\LL$ is  positive semi-definite. Moreover, for a nonnegative diagonal matrix $\XX$ with at least  one nonzero diagonal entry, we have the fact that every element of $\kh{\LL + \XX}^{-1}$ is positive~\cite{McNeSc95,LiPeShYiZh19,MaAb19}.

\section{Problem Formulation}

In this section, we formulate the problem for optimizing the opinion influence of leaders in a social network described by a graph $\calG(V, E)$, where  nodes represent  individuals or agents, and edges represent interactions among  agents. We adopt the discrete-time leader-follower DeGroot model~\cite{MeAsDaAmAn13,VaFaFr14,MaPa19} for opinion dynamics to achieve our goal.



\subsection{Leader-Follower DeGroot Model}

In the leader-follower DeGroot model on graph $\calG$,  node set is divided into two disjoint parts, a set $S$ of a small number  $s \ll n$ leader nodes  and a follower set $F$ containing the remaining  $n -s$ nodes, where $S$ is the union of  two disjoint subsets $S_0$ and $ S_1$ presenting two competing parties and satisfying $S = S_0 \cup S_1$. Each node has a nonnegative scalar-value opinion belongs to  interval $[0,1]$. Let  $\xx_i(t)$ be  the  opinion of node $i$ at time $t$. If $i\in S_0$, $\xx_i(t)$ keeps unchanged, meaning  $\xx_i(t)=0$ for all $t$;   if $i\in S_1$, $\xx_i(t)=1$  for all $t$. For a follower node $i \in F$  with an initial opinion $\xx_i(0)$, it will update its opinion by averaging all its neighbours' opinions as
\begin{equation}\label{eq:dym}
    \xx_i(t+1) = \frac{\sum_{j \in \calN_i}\AA_{i,j}\xx_j(t)}{\sum_{j \in \calN_i}\AA_{i,j}}.
\end{equation}
Let $\xx_F(t)$ and $\xx_S(t)$ be, respectively, opinion vectors of followers and leaders at time $t$.  For large $t$, $\xx_F(t)$ and $\xx_S(t)$ converge. Let  $\xx_S(\infty)=\lim_{t\rightarrow \infty}\xx_S(t)$ and  $\xx_F(\infty)=\lim_{t\rightarrow \infty}\xx_F(t)$. We write   matrices $\AA$ and $\LL$ in block form as
\begin{align*}
	\AA = \begin{pmatrix}
	\AA_{S, S} & \AA_{S, F}\\
	\AA_{F, S} & \AA_{F, F}
	\end{pmatrix} ,\,
	\LL = \begin{pmatrix}
	\LL_{S, S} & \LL_{S, F}\\
	\LL_{F, S} & \LL_{F, F}
	\end{pmatrix}\, .
	\end{align*}
Then, the  stationary opinions can be determined as follows~\cite{YiCaPa19}:
\begin{align}
  \xx_S(\infty) =& \xx_S(0), \\
  \xx_F(\infty) =& \LL_{F,F}^{-1}\AA_{F,S}\xx_S(0)=\LL_{F}^{-1}\AA_{F,S}\xx_S(0).\label{eq:dymF}
\end{align}

\subsection{Problem Statement}\label{ProbStat}

Equation~\eqref{eq:dymF} shows that for any follower node, its equilibrium opinion is determined by the opinions of leaders, independent of its own initial opinion.  Thus, the leaders affect the opinions of followers. Let $H(\calG)$ denote the sum of stable opinions over all followers in graph $\calG$, given by $H(\calG)=\sum_{i \in F} \xx_i(\infty) = \boldsymbol{1}^{\top}\LL_{F}^{-1} \AA_{F,S} \xx_S(0)$, which encodes the influence of leaders on the opinions of followers.  In the sequel, by using the quantity  $H(\calG)$,  we formulate the problem for maximizing the influence of the $1$-leader by adding a fixed number of edges from a candidate set of edges. For simplicity, we define a vector $\bb = \AA_{F,S} \xx_S(0)$. By definition, for each node $i \in F$, $\bb_i$ is in fact equal to the number of edges connecting node $i$ and nodes in  set $S_1$.  Then, the vector for equilibrium opinions of followers can be  simplified as
$\xx_F(\infty) = \LL_{F}^{-1} \bb$ and $H(\calG)$ can be rewritten as $H(\calG)=\boldsymbol{1}^{\top}\LL_{F}^{-1} \bb$.

For a connected undirected graph $\calG(V,E)$, if we add a set $T$ of  nonexistent edges from a candidate set $Q$ to $\calG$ forming a new graph $\calG(T)=(V,E \cup T)$, where each new edge  connects a 1-leader node in $S_1$ and a follower node in $F$, the overall equilibrium opinion of follower nodes in $\calG(T)=(V,E \cup T)$  will increase. We next prove this property. For simplicity, we simplify $H(\calG(T))$ as $H(T)$, which means  $H(\calG)=H(\emptyset)$.  Then, we have the following result.

\begin{lemma}\label{lem:update}
Let $\calG=(V,E)$ be a connected graph with nonempty follower set $F$,  0-leader set $S_0$, and 1-leader set $S_1$. Let $e \notin E$ be a potential  edge incident to node $i \in F$ and a 1-leader in $S_1$. Define
$\Delta(e) \defeq H(\{e\})-H(\emptyset)$. Then,
\begin{equation}\label{Eq:inc}
  \Delta(e) = \frac{\boldsymbol{1}^\top \LL_F^{-1} \ee_i(1-\ee_i^\top \LL_F^{-1} \bb)}{1+\ee_i^\top\LL_F^{-1}\ee_i}
\end{equation}
and $\Delta(e)\geq 0$.
\end{lemma}
\begin{proof}
By definition, one obtains $ H(\{e\})= H(\calG(\{e\}))=\boldsymbol{1}^\top (\LL_F+\ee_i\ee_i^\top)^{-1}(\bb+\ee_i)$.  Exploiting Sherman-Morrison formula~\cite{Me73}, it follows that
\begin{equation*}
\left( \LL_{F} + \ee_i\ee_i^\top \right)^{-1} = \LL_{F}^{-1} -  \frac{\LL_{F}^{-1}\ee_i\ee_i^\top\LL_{F}^{-1}}{1+\ee^\top_i\LL_{F}^{-1}\ee_i}.
\end{equation*}
Then, $\Delta(e)$ is evaluated as
\begin{align*}
   \Delta(e) &=\boldsymbol{1}^\top\left(\LL_{F}^{-1} - \frac{\LL_{F}^{-1}\ee_i\ee_i^\top\LL_{F}^{-1}}{1+\ee^\top_i\LL_{F}^{-1}\ee_i}\right)(\bb+\ee_i) -\boldsymbol{1}^\top\LL_F^{-1}\bb \\
  & = \frac{\boldsymbol{1}^\top \LL_F^{-1} \ee_i(1-\ee_i^\top \LL_F^{-1} \bb)}{1+\ee_i^\top\LL_F^{-1}\ee_i}.
\end{align*}
Notice that $\ee_i^\top \LL_F^{-1} \bb$ is the equilibrium opinion of follower node $i$ in graph $\calG$, which is no more than 1. On the other hand, any entry of matrix $\LL_{F}^{-1}$ is nonnegative. Hence, $\Delta(e) \geq 0$.
\end{proof}

Lemma~\ref{lem:update} indicates that the addition of any nonexisting  edge  connecting a 1-leader and a follower will lead to an increase of  the overall equilibrium opinion of followers.  Then we naturally raise the following  problem called \textsc{OpinionMaximization}: How to optimally select a set $T$ with $k$ edge in a candidate edge set $Q$,  so that the influence of 1-leaders quantified by the overall opinion of leaders in the new graph is maximized.  Mathematically, the influence maximization problem   can be formally stated  as follows.



\begin{problem}[Opinion Maximization]\label{prob:om}
Given a connected undirected graph $\calG=(V, E)$, a nonempty  set $S_0$ of  0-valued leaders, a set $S_1 \ne \emptyset$ of 1-leaders, a nonempty  set $F=V \setminus (S_0 \cup S_1)$ of followers, the candidate edge set $Q \subseteq F\times S_1$ consisting of nonexistent edges connecting 1-leaders and follower nodes, and an integer $k$, we aim to find the edge set $T \subseteq Q$ with $|T|= k$, and add these chosen $k$ edges to graph $\calG$ forming a new graph $\calG(T)=(V,E \cup T)$, so that the overall opinion $H(T)$ is maximized. This set optimization problem can be formulated as:
\begin{equation}\label{Pro:2}
  T = \arg \max_{P \subseteq Q, |P|= k} H(P).
\end{equation}
\end{problem}


Similarly, we can maximize the influence of  0-valued leaders by adding  edges to  graph $\calG$ to minimize the overall opinion of followers, which   is  called \textsc{OpinionMinimization} problem.  Since for both problems, the proof and algorithms are similar, in what follows, we only consider  the  \textsc{OpinionMaximization} problem.

Note that given set $S_1$ of 1-leaders,  the number of absent edges connecting 1-leaders and followers in a sparse network is large with order $O(|S_1|n)$.  However,  if the stationary  opinion of follower node $i$ is close to 1, the benefit of adding a new edge between $i$ and a node in set $S_1$ is relatively low. Thus, we will not   inspect all of $O(|S_1|n)$ nonexistent edges,  but  focus on a small number of candidate edges with $|Q| \ll |S_1|n$.  Specifically, we set a threshold  $\eta$  ($1/2 < \eta < 1$) and determine  the candidate set $Q$ according to the  following rule. For each follower node $i$ not adjacent to  any 1-leader  in  graph $\calG$, if $\xx_i(\infty) < \eta$,  $Q$ includes the nonexistent edges connecting $i$ and 1-leaders. Similar restriction to the candidate set $Q$ of edges is previously used in~\cite{AmSi19}, where the edges in $Q$ are called ``good'' candidate edges.

\subsection{Properties of Objective Function}

Here  we prove that as the objective function of  Problem~\ref{prob:om}, the set function  $H(\cdot)$ has two desirable properties:  monotonicity and submodularity. First, by Lemma~\ref{lem:update}, it is immediate that function $H(\cdot)$ is  monotonically increasing.




\begin{theorem}(Monotonicity) \label{thm:MI}
For two subsets $B$ and $ T$  of edges satisfying $B\subset T \subset Q$,  $H(B) \leq H(T)$ holds.
\end{theorem}
Next, we  show that  function $H(\cdot)$ is  submodular.
\begin{theorem}(Submodularity) \label{thm:SM}
For  two subsets $B$ and $ T$  obeying   $B \subset T \subseteq Q$ and any edge $e \in Q \setminus T$ ,
	\begin{align}\label{thm:SMeq}
	H(T\cup \{ e \}) - H(T ) \leq H(B\cup \{ e \}) - H(B).
	\end{align}
\end{theorem}
\begin{proof}
Let $e_1=(a,i)$ and $e_2=(b,j)$ be two edges in $Q$ with $a,b \in S_1$ and $i,j \in F$. We now prove that for this simple case, one has
\begin{equation}\label{eq:12}
  H(\{e_1 \cup e_2\}) - H( \{e_1\}) \leq H( \{e_2\}) - H(\emptyset ).
\end{equation}
To this end, we define matrix $\Omega(x,y) = (\LL_{F}+x\EE_{ii}+y\EE_{jj})^{-1}$ and vector $\bar{\bb}(x,y) = \bb + x\ee_i + y\ee_j$, where $\EE_{ii} = \ee_i\ee_i^\top$, $x \geq 0$ and $y \geq 0$. By definition, it is easy to verify  that the entries in $\Omega(x,y)$ are nonnegative, and inequality~\eqref{eq:12} can be rephrased as
\begin{align}\label{eq:01}
	& \boldsymbol{1}^\top\Omega(1,1)\bar{\bb}(1,1)-\boldsymbol{1}^\top\Omega(0,1)\bar{\bb}(0,1)\notag\\
	 \leq & \boldsymbol{1}^\top\Omega(1,0)\bar{\bb}(1,0)-\boldsymbol{1}^\top\Omega(0,0)\bar{\bb}(0,0).
\end{align}

In order to prove~\eqref{eq:01}, we introduce a function $f(x,y) \defeq \boldsymbol{1}^\top\Omega(x,y)$ $\bar{\bb}(x,y)$. Next we prove
\begin{equation}\label{eq:xy}
f(x,y)-f(0,y)\leq f(x,0)-f(0,0),
\end{equation}
the special case $x=y=1$ of which is exactly~\eqref{eq:01}.  In order to  prove~\eqref{eq:xy}, it suffices to prove
$f_{xy}(x,y) \leq 0$, the proof of which involves the following  matrix derivative formula
\begin{align*}
	\frac{d}{dt}\AA(t)^{-1} = -\AA(t)^{-1} \frac{d}{dt}\AA(t) \AA(t)^{-1}.
\end{align*}
Since $H(\cdot)$ is a monotone increasing function, we have
\begin{align}
f_x(x,y)=-\boldsymbol{1}^\top\Omega(x,y)\EE_{ii}\Omega(x,y)\bar{\bb}(x,y)+\boldsymbol{1}^\top\Omega(x,y)\ee_i\geq 0,\notag
 \end{align}
\begin{align}
f_y(x,y)=-\boldsymbol{1}^\top\Omega(x,y)\EE_{jj}\Omega(x,y)\bar{\bb}(x,y)+\boldsymbol{1}^\top\Omega(x,y)\ee_j\geq 0.\notag
\end{align}
Then,  $f_{xy}(x,y)$ can be computed as
\begin{align}\label{eq:2xy}
 f_{xy}(x,y)
    &= \boldsymbol{1}^\top\Omega(x,y)\EE_{jj}\Omega(x,y)\EE_{ii}\Omega(x,y)\bar{\bb}(x,y)+ \notag\\
    & \quad\boldsymbol{1}^\top\Omega(x,y)\EE_{ii}\Omega(x,y) \EE_{jj}\Omega(x,y)\bar{\bb}(x,y)- \notag\\
    & \quad \boldsymbol{1}^\top\Omega(x,y)\EE_{ii}\Omega(x,y)\ee_j -\boldsymbol{1}^\top\Omega(x,y)\EE_{jj}\Omega(x,y)\ee_i \notag\\
    &=-\ee_i^\top \Omega(x,y) \ee_j (f_x(x,y)+f_y(x,y)) \leq 0.
\end{align}
Combining~\eqref{eq:01},~\eqref{eq:xy} and~\eqref{eq:2xy} yields~\eqref{eq:12}.

We assume  $|T \setminus B|=z$ and $T \setminus B=\{e_1, e_2,\ldots, e_z\}$, and  define  a graph sequence $\calG^{(i)}$ ($i=0,1, 2, \ldots,  z$) with identical node set $V$ but different  edge set $E^{(i)}$, obeying  $\calG^{(0)}=\calG$, $E^{(0)}=E$, and $E^{(i)} \setminus E^{(i-1)}=e_i$. By iteratively applying~\eqref{eq:12} to $\calG^{(i)}$ leads to~\eqref{thm:SMeq}, we can conclude the proof of  submodularity.
\end{proof}

\section{Simple Greedy Algorithm}

Problem~\ref{prob:om} is inherently a combinatorial problem. It can be solved by the following na\"{\i}ve brute-force approach. For each set $T$ of the $\tbinom{|Q|}{k}$ possible subsets  of edges, compute the overall equilibrium opinion of follower nodes in the resultant graph when all edges in this set are added. Then, output the subset $T^*$ of edges, whose addition leads to the largest increase for the overall opinion of  followers. Although this method is simple, it is computationally impossible even for small networks, since its has an exponential complexity $O\big(\tbinom{|Q|}{k}\cdot (n-s)^3\big)$.

To tackle the exponential complexity, one often resorts to greedy heuristics.  Due to the monotonicity and submodularity of   the objective function $H(\cdot)$, a simple greedy strategy is guaranteed  to have a $(1-\frac{1}{e})$-approximation solution to  Problem~\ref{prob:om}~\cite{NeWoFi78}. Initially, we set the edge set $T$ to be empty, then $k$ edges are added from set $Q \setminus T$ iteratively. In each iteration step $i$, edge $e_i$ in set $Q$ of candidate edges  is selected, which  maximizes the overall opinion of followers. The algorithm terminates when $k$ edges are selected to be added to $T$. For every candidate edge, it needs computation of the overall opinion,
which involves matrix inversion. A direct calculation of  matrix inversion requires $O((n-s)^3)$ time, leading to a total computation complexity $O(k|Q|(n-s)^3)$.

Actually, as shown in the proof of Lemma~\ref{lem:update}, if $\LL_F^{-1}$ is already computed, adding a single edge can be
looked upon as a  rank-1 update to matrix $\LL_F^{-1}$, which can be done in  time $O((n-s)^2)$ by using the Sherman-Morrison formula~\cite{Me73},  rather than directly inverting a matrix in time $O((n-s)^3)$. This leads to our simple algorithm \textsc{Exact}$(G, S_1, S_0, Q, k)$ to solve Problem~\ref{prob:om}, which is  outlined in Algorithm \ref{alg:1}. The algorithm first calculates the inversion of matrix $\LL_F$ in time $O((n-s)^3)$. Then it works in $k$ rounds with each round containing two main operations. One is to compute $\Delta(e)$ (Line 4) in $O(n(n-s))$ time, the other is to update $\LL_F^{-1}$ in $O((n-s)^2)$ time when  a new edge $e_i$ (Line 8) is added. Therefore, the whole running time of Algorithm \ref{alg:1} is $O((n-s)^3+kn(n-s))$,  much faster than the  brute-force algorithm.

\begin{algorithm}[tb]
	\caption{\textsc{Exact}$(\calG, S_1, S_0, Q, k)$}
    \label{alg:1}
	\Input{
		A connected graph $\calG=(V,E)$;
        two disjoint node sets $S_0, S_1$;
        a candidate edge set $Q$;
        an integer $1 \leq k \leq |Q|$\\
	}
	\Output{
		$T$: A subset of $Q$ with $|T|=k$
	}
Compute $\LL_F^{-1}$ and $\bb$\;
Initialize solution $T = \emptyset$ \;
\For{$i = 1$ to $k$}{
Compute $\Delta(e)$ for each $e \in Q \setminus T$\;
Select $e_i$ s. t.  $e_i \gets \mathrm{arg\, max}_{e \in Q \setminus T} \Delta(e)$ and $e_i$ is incident to node $j\in F$\;
Update solution $T \gets T \cup \{ e_i \}$ \;
Update the graph $\calG \gets \calG(V, E \cup \{ e_i \})$ \;
Update $\LL_F^{-1} \gets \LL_{F}^{-1} -  \frac{\LL_{F}^{-1}\ee_j\ee_j^\top\LL_{F}^{-1}}{1+\ee_j^\top\LL_{F}^{-1}\ee_j}$\;
Update $\bb \gets \bb+\ee_j$

}
\Return $T$.
\end{algorithm}

Based on the well-established result~\cite{NeWoFi78}, Algorithm \ref{alg:1} yields a $(1-\frac{1}{e})$-approximation of the optimal solution to Problem \ref{prob:om}, as provided in the following theorem.
\begin{theorem}
The set $T$ returned by Algorithm \ref{alg:1} satisfies
\begin{equation*}
H(T) -H(\emptyset) \geq \left(1-\frac{1}{e}\right) \big(H(T^*) -H(\emptyset) \big),
\end{equation*}
where $T^*$ is the optimal solution to  Problem \ref{prob:om} satisfying
\begin{equation*}
 H(T^*) =\argmax_{P \subset Q,|P|=k} H(P).
\end{equation*}
\end{theorem}
%
%
%
%

\section{Fast Greedy Algorithm}

Although the computation complexity of  Algorithm \ref{alg:1} is much lower than the brute-force algorithm, it  still cannot handle  large-scale networks since it requires inverting  matrix $\LL_F$ in cube time. As shown above, the key step to solve Problem \ref{prob:om} is to compute the quantity $\Delta(e)$ given in~\eqref{Eq:inc}, which is the increment of overall opinion resulted by the addition of edge $e$. In this section, we provide efficient approximations for the three terms $\boldsymbol{1}^\top \LL_F^{-1} \ee_i$, $1-\ee_i^\top \LL_F^{-1} \bb$, and $\ee_i^\top\LL_F^{-1}\ee_i$ in the numerator and denominator of~\eqref{Eq:inc}, which avoid inverting a matrix. These approximations together lead to an error-guaranteed approximation to $\Delta(e)$ and thus a fast approximation algorithm to Problem \ref{prob:om} that has $(1-\frac{1}{e}-\epsilon)$ approximation ratio and time complexity $\tilde{O}(km\epsilon^{-2})$ for any error parameter $0< \eps < 1/2$.

In order to reduce computational cost, it is important to avoid the inversion of matrix  $\LL_F$, which can be done since  matrix  $\LL_F$ has many desirable properties as will be shown below. For example,  it can be expressed as the  sum of a Laplacian matrix associated with a graph  and a nonnegative diagonal matrix. For  a connected graph $\calG=(V,E)$ with a subset $F\subset V$ of follower nodes and a subset $S \subset V$ of leader nodes,  we define a graph $\bar{\calG}=(F,(F\times F) \cap E)$ with $n-s$ nodes and  $\bar{m}$ edges, which is a subgraph of  $\calG$. Let $\bar{\BB}$ be the incidence matrix of $\bar{\calG}$ and $\bar{\LL}$ be its Laplacian matrix. Then, matrix $\LL_F$ can be represented as $\LL_F= \bar{\BB}^\top \bar{\BB} + \WW$, where $\WW$ is a nonnegative diagonal matrix with  the $i^{\rm th}$  diagonal entry equal to the number of edges connecting the follower node $i$ and nodes in $S$. Thus, $\LL_F$ is a symmetric, diagonally-dominant M-matrix (SDDM). Moreover, it is easy to obtain that $\ee_i^\top \LL_{F}^{-1} \ee_i \geq 1/n$ and $\trace{\LL_{F}^{-1}} \leq \trace{\bar{\LL}^{-1}}\leq(n^2-1)/6$~\cite{LoLo03}. These properties are helpful for the following derivations.


In addition, the following lemma is also instrumental for approximating the related quantities.
\begin{lemma}\label{lem:2norm}
For an arbitrary vector $\vvv$ of $n-s$ dimension,  $\vvv_i^2 \leq n \norm{\vvv}_{\LL_{F}}^2$ holds, where $\norm{\vvv}_{\LL_{F}}^2=\vvv^\top \LL_F \vvv$.
\end{lemma}
\begin{proof}
Considering $\LL_F= \bar{\BB}^\top \bar{\BB} + \WW$, we distinguish two cases:  $\WW_{i,i}\geq 1$ and $\WW_{i,i}= 0$. For the first case $\WW_{i,i}\geq 1$,  it is apparent that $\vvv_i^2 \leq n||\vvv||_{\LL_F}^2$. While for the second case $\WW_{i,i} = 0$,  there exists a follower node $j$ with corresponding  element $\WW_{j,j} \geq 1$ in the component of graph $\bar{\calG}$ that contains  node $j$. Let $P_{ij}$ be a simple path connecting node $i$ and $j$ in  graph $\bar{\calG}$, the length of  which  is at most $n$. Then, we have
\begin{align*}
  ||\vvv||_{\LL_F}^2 &\geq \sum_{(a,b) \in P_{ij}} (\vvv_a-\vvv_b)^2 +\vvv_j^2 \\
  & \geq \frac{(\sum_{(a,b) \in P_{ij}}(\vvv_a-\vvv_b)+\vvv_j)^2}{n} \geq \frac{\vvv_i^2}{n},
\end{align*}
which completes the proof.
\end{proof}

\subsection{Approximations of  Numerator in~\eqref{Eq:inc} }%



We now approximate the two terms $\one^\top \LL_F^{-1} \ee_i$, $1-\ee_i^\top \LL_F^{-1} \bb$ in the numerator in~\eqref{Eq:inc}. Since $\LL_F$ is a SDDM matrix, one can resort to the fast symmetric, diagonally-dominant (SDD) linear system solver~\cite{SpTe14,CoKyMiPaJaPeRaXu14}  to evaluate $\boldsymbol{1}^\top \LL_F^{-1} \ee_i$ and $1-\ee_i^\top \LL_F^{-1} \bb$, which  avoids inverting matrix $\LL_F$.
\begin{lemma}\label{lem:solver}
There is a nearly linear time solver $\xx = \SDDMSolver(\SS, \yy, \eps)$ which takes a symmetric positive semi-definite matrix $\SS_{n\times n}$ with $m$ nonzero entries, a vector $\bb \in \mathbb{R}^n$, and an error parameter $\delta > 0$, and returns a vector $\xx \in \mathbb{R}^n$ satisfying $\norm{\xx - \SS^{-1} \yy}_{\SS} \leq \delta \norm{\SS^{-1} \yy}_{\SS}$ with high probability, where $\norm{\xx}_{\SS} \defeq \sqrt{\xx^\top \SS \xx}$. The solver runs in expected time $\tilde{O}(m)$, where $\tilde{O}(\cdot)$ notation suppresses the ${\rm poly} (\log n)$ factors. 
\end{lemma}

Based on this solver, $\boldsymbol{1}^\top \LL_F^{-1} \ee_i$, $1-\ee_i^\top \LL_F^{-1} \bb$ are approximated in Lemmas~\ref{lem:num1} and~\ref{lem:num2}, respectively.
\begin{lemma}\label{lem:num1}
Given an undirected unweighted graph $\calG=(V,E)$, the matrix $\LL_{F}$ and a parameter $0<\epsilon < 1/2$, let $\hh = \SDDMSolver(\LL_F, \boldsymbol{1}, \delta_1)$ where $\delta_1 = \frac{\epsilon}{2n^2\sqrt{6(n^2-1)}}$. Then for any $i \in F$, we have
\begin{equation}
\boldsymbol{1}^\top \LL_F^{-1} \ee_i \mathop{\approx}\limits^{\epsilon/6} \hh_i.
\end{equation}
\end{lemma}

\begin{proof}
Define $\tilde{\hh} =  \LL_F^{-1}\boldsymbol{1}$, then $\tilde{\hh}_i = \boldsymbol{1}^\top \LL_F^{-1} \ee_i$. According to Lemmas~\ref{lem:2norm} and~\ref{lem:solver}, one obtains
\begin{align*}
    (\hh_i-\tilde{\hh}_i)^2 \leq&  n \smallnorm{\hh-\tilde{\hh}}_{\LL_F}^2
\leq \delta_1^2 n \smallnorm{\tilde{\hh}}^2_{\LL_F}  \\
\leq & \delta_1^2 n^2 \trace{\LL_F^{-1}} \leq \delta_1^2 n^2(n^2-1)/6.
\end{align*}
On the other hand, $\tilde{\hh}_i$ can be bounded as
\begin{equation*}
\tilde{\hh}_i =  \boldsymbol{1}^\top \LL_F^{-1} \ee_i \geq \ee_i^\top \LL_F^{-1} \ee_i  \geq \frac{1}{2n}.
\end{equation*}
Then, one has
\begin{equation*}
  \frac{|\hh_i-\tilde{\hh}_i|}{\tilde{\hh}_i} \leq 2\delta_1 n^2 \sqrt{(n^2-1)/6} \leq \frac{\epsilon}{6},
\end{equation*}
which completes the proof.
\end{proof}
\begin{lemma}\label{lem:num2}
Given an undirected unweighted graph $\calG=(V,E)$, the matrix $\LL_{F}$ and two parameters $0<\epsilon < 1/2$, $1/2 < \eta <1$, let $\pp = \SDDMSolver(\LL_F, \bb, \delta_2)$ where $\delta_2=\frac{(1-\eta)\epsilon}{n^2\sqrt{6(n^2-1)}}$. Then for any $i \in F$, we have
\begin{equation}
1-\ee_i^\top \LL_F^{-1} \bb \mathop{\approx}\limits^{\epsilon/6} 1-\pp_i.
\end{equation}
\end{lemma}

\begin{proof}
Define $\tilde{\pp} = \LL_F^{-1}\bb$, then $\tilde{\pp}_i = \ee_i^\top \LL_F^{-1} \bb$. Considering  the restriction of the set $Q$ of candidate edges, for each  follower $i$ connecting by an $e \in Q$,  $\ee_i^\top\LL_F^{-1}\bb \leq \eta$ holds. By Lemmas~\ref{lem:2norm} and~\ref{lem:solver}, we have
\begin{equation*}
|(1-\tilde{\pp}_i)-(1-\pp_i)|=|\tilde{\pp}_i - \pp_i| \leq  \delta_2 n^2\sqrt{(n^2-1)/6}.
\end{equation*}
Therefore,
\begin{equation*}
  \frac{|(1-\pp_i)-(1-\tilde{\pp}_i)|}{1-\tilde{\pp}_i} \leq \frac{\delta_2 n^2\sqrt{(n^2-1)/6}}{1-\eta}  \leq \frac{\epsilon}{6},
\end{equation*}
which finishes the proof.
\end{proof}

\subsection{Approximation of  Denominator~\eqref{Eq:inc}}


Using the expression $\LL_F =\bar{\BB}^\top\bar{\BB}+\WW$, $ \ee_i^\top \LL_{F}^{-1} \ee_i$ is  recast as
\begin{align}
  & \ee_i^\top \LL_{F}^{-1} \ee_i
  =  \ee_i^\top \LL_{F}^{-1}(\bar{\BB}^\top\bar{\BB} + \WW)\LL_{F}^{-1} \ee_i \notag\\
  = & \ee_i^\top \LL_{F}^{-1}\bar{\BB}^\top\bar{\BB}\LL_{F}^{-1} \ee_i + \ee_i^\top \LL_{F}^{-1}\WW\LL_{F}^{-1} \ee_i \notag\\
  = & \smallnorm{\bar{\BB} \LL_{F}^{-1} \ee_i}^2 + \smallnorm{\WW^{\frac{1}{2}} \LL_{F}^{-1} \ee_i}^2.
\end{align}
In this way, we have reduced the estimation of $\ee_i^\top \LL_{F}^{-1} \ee_i$ to the calculation of the $\ell_2$ norms $\smallnorm{\bar{\BB} \LL_F^{-1} \ee_i}^2$ and $\smallnorm{\WW^{\frac{1}{2}} \LL_F^{-1} \ee_i}^2$ of two vectors in $\mathbb{R}^{\bar{m}}$ and $\mathbb{R}^{n-s}$, respectively. However, the complexity for exactly computing these two $\ell_2$ norms is still high. In order to alleviate the computation burden,   we apply the Johnson-Lindenstrauss (JL) Lemma~\cite{JoLi84, Ac01} to approximate the $\ell_2$ norms. For example, for $\smallnorm{\bar{\BB} \LL_F^{-1} \ee_i}^2$, if we project a set of $n-s$ vectors of $\bar{m}$ dimension  (like the columns of matrix $\bar{\BB} \LL_F^{-1}$) onto a low $t$-dimensional subspace spanned by the columns of a random matrix $\QQ \in \mathbb{R}^{t\times \bar{m}}$ with entries being $\pm 1/\sqrt{t}$, where $t=\ceil{24\log (n-s)/ \eps^2}$ for given $\epsilon$, then we  get an $\epsilon$-approximation of $\norm{\bar{\BB} \LL_F^{-1} \ee_i}^2$ with high probability.  This $\ell_2$ norm preserving projection method is also applicable to the estimation of $\smallnorm{\WW^{\frac{1}{2}} \LL_F^{-1} \ee_i}^2$. For consistency, we introduce the JL Lemma~\cite{JoLi84, Ac01}.



\begin{lemma}
\label{lemma:JL}
	Given fixed vectors $\vvv_1,\vvv_2,\ldots,\vvv_n\in \mathbb{R}^d$ and
	$\epsilon>0$, let
 $\QQ_{t\times d}$, $t\ge 24\log n/\epsilon^2$, be a matrix,    each entry of which is equal to $1/\sqrt{t}$ or $-1/\sqrt{t}$ with the same probability $1/2$.
	 Then with probability at least $1-1/n$,
	\[(1-\epsilon)\|\vvv_i-\vvv_j\|^2\le \|\QQ \vvv_i-\QQ \vvv_j\|^2\le
	(1+\epsilon)\|\vvv_i-\vvv_j\|^2\] for all pairs $i,j\le n$.
\end{lemma}

Let $\QQ_{t \times \bar{m}}$ and $\PP_{t \times (n-s)}$ be two random $\pm 1/\sqrt{t}$ matrices where $t = \ceil{24\log (n-s)/ \eps^2}$.  By Lemma~\ref{lemma:JL}, for any $ i\in F$ we have
	\begin{align*}
	& \ee_i^\top \LL_F^{-1} \ee_i
	\mathop{\approx}\limits^{\epsilon}
	\smallnorm{\QQ \bar{\BB} \LL_F^{-1} \ee_i}^2 +
	\smallnorm{\PP \WW^{\frac{1}{2}} \LL_F^{-1} \ee_i}^2.
	\end{align*}
Let $\bar{\XX}=\QQ \bBB$, $\XX=\bBB \LL_{F}^{-1}$, $\XX'=\QQ\bBB \LL_{F}^{-1}$, $\tilde{\XX}_{[j, :]} = \SDDMSolver$ $(\LL_{F}, \bar{\XX}_{[j, :]}, \delta)$ and $\bar{\YY} = \PP \WW^{\frac{1}{2}}$, $\YY=\WW^{\frac{1}{2}}\LL_{F}^{-1}$, $\YY'=\PP \WW^{\frac{1}{2}}\LL_{F}^{-1}$, $\tilde{\YY}_{[j, :]} = \SDDMSolver(\LL_{F}, \bar{\YY}_{[j, :]},$ $ \delta)$. Then, $\ee_i^\top \LL_F^{-1} \ee_i=\smallnorm{\XX\ee_i}^2+\smallnorm{\YY\ee_i}^2$. Combining Lemmas~\ref{lem:solver} and~\ref{lemma:JL}, we get an approximation of $\ee_i^\top \LL_F^{-1} \ee_i$.

\begin{lemma}\label{lem:dom}
Suppose that $\forall i \in F$,
\begin{align}
   & (1-\epsilon/12)\smallnorm{\XX\ee_i}^2 \leq \smallnorm{\XX'\ee_i}^2 \leq  (1+\epsilon/12)\smallnorm{\XX\ee_i}^2, \notag\\
   & (1-\epsilon/12)\smallnorm{\YY\ee_i}^2 \leq \smallnorm{\YY'\ee_i}^2 \leq  (1+\epsilon/12)\smallnorm{\YY\ee_i}^2,\notag
\end{align}
and $\forall j, 1\leq j \leq t$,
\begin{align}
   & \smallnorm{\XX'_{[j, :]}-\tilde{\XX}_{[j, :]}}_{\LL_{F}} \leq \delta \smallnorm{\XX'_{[j, :]}}_{\LL_{F}}, \notag\\
   & \smallnorm{\YY'_{[j, :]}-\tilde{\YY}_{[j, :]}}_{\LL_{F}} \leq \delta \smallnorm{\YY'_{[j, :]}}_{\LL_{F}},\notag
\end{align}
where $\delta< \frac{\epsilon}{72n^2}\sqrt{\frac{6(1-\epsilon/12)}{(1+\epsilon/12)(n^2-1)}}$, we then have
\begin{equation}\label{eiLLei}
e_i^\top \LL_F^{-1} \ee_i= \smallnorm{\XX\ee_i}^2 + \smallnorm{\YY\ee_i}^2 \mathop{\approx}\limits^{\epsilon/3} \smallnorm{\tilde{\XX}\ee_i}^2+\smallnorm{\tilde{\YY}\ee_i}^2.
\end{equation}
\end{lemma}

\begin{proof}
According to triangle inequality, one has
\begin{align*}
       & \left| \smallnorm{\tilde{\XX}\ee_i} - \smallnorm{\XX'\ee_i} \right|  \leq \smallnorm{(\tilde{\XX}-\XX')\ee_i} \leq \smallnorm{\tilde{\XX}-\XX'}_F  \\
     = & \sqrt{\sum_{j=1}^{t} \norm{\tilde{\XX}_{[j, :]} - \XX'_{[j, :]}}^2} \leq \sqrt{\sum_{j=1}^{t} n^2\norm{\tilde{\XX}_{[j, :]} - \XX'_{[j, :]}}_{\LL_{F}}^2} \\
  \leq & \sqrt{\sum_{j=1}^{t} \delta^2 n^2 \norm{\XX'_{[j, :]}}_{\LL_{F}}^2} \leq \delta n\sqrt{n}\smallnorm{\XX'}_F\\
  \leq & \delta n \sqrt{n\sum_{i \in F} (1+\epsilon/12)\smallnorm{\XX\ee_i}^2} \\
  \leq & \delta n \sqrt{(1+\epsilon/12)n\sum_{i \in F} \ee_i^\top \LL_{F}^{-1} \ee_i} \\
  \leq & \delta n  \sqrt{(1+\epsilon/12)n(n^2-1)/6}.
\end{align*}
Similarly, for $\smallnorm{\YY'\ee_i}$ and $\smallnorm{\tilde{\YY}\ee_i}$, one  obtains
\begin{equation}
   \left| \smallnorm{\tilde{\YY}\ee_i} - \smallnorm{\YY'\ee_i} \right| \leq \delta n \sqrt{(1+\epsilon/12)n(n^2-1)/6}. \notag
\end{equation}


On the other hand,
\begin{align*}
  & \norm{\XX'\ee_i}^2 + \norm{\YY'\ee_i}^2 \geq (1-\epsilon/12) (\norm{\XX\ee_i}^2 + \norm{\YY\ee_i}^2) \\
  \geq & (1-\epsilon/12)\ee_i^\top \LL_F^{-1} \ee_i \geq (1-\epsilon/12)/2n.
\end{align*}
For $\norm{\XX'\ee_i}$ and $\norm{\YY'\ee_i}$, at least one is no less than $\sqrt{\frac{1-\epsilon/12}{4n}}$.  We only consider the case  $\smallnorm{\XX'\ee_i} \geq \sqrt{\frac{1-\epsilon/12}{4n}}$, since the $\smallnorm{\YY'\ee_i} \geq \sqrt{\frac{1-\epsilon/12}{4n}}$ can be handled in a similar way.

For  $\smallnorm{\XX'\ee_i} \geq \sqrt{\frac{1-\epsilon/12}{4n}}$, we distinguish two cases: (i) $\smallnorm{\YY'\ee_i} > \smallnorm{\XX'\ee_i} \geq  \sqrt{\frac{1-\epsilon/12}{4n}}$,  (ii)   $\smallnorm{\XX'\ee_i} \geq \smallnorm{\YY'\ee_i}$.
For case (i), one has
\begin{equation}
  \frac{\left|\smallnorm{\XX'\ee_i}- \smallnorm{\tilde{\XX}\ee_i}\right|}{\smallnorm{\XX'\ee_i}} \leq  2\delta n^2 \sqrt{\frac{(1+\epsilon/12)(n^2-1)}{6(1-\epsilon/12)}} \leq \frac{\epsilon}{36},\notag
\end{equation}
Based on this result, one  further obtains
\begin{align*}
&\big|\smallnorm{\XX'\ee_i}^2- \smallnorm{\tilde{\XX}\ee_i}^2\big|  \\
= &\big|\smallnorm{\XX'\ee_i}- \smallnorm{\tilde{\XX}\ee_i}\big|\cdot\big|\smallnorm{\XX'\ee_i}+ \smallnorm{\tilde{\XX}\ee_i}\big| \\
\leq & \frac{\epsilon}{36}\left (2+\frac{\epsilon}{36}\right) \smallnorm{\XX'\ee_i}^2
\leq  \frac{\epsilon}{12} \smallnorm{\XX'\ee_i}^2 \leq \frac{\epsilon}{6} \smallnorm{\XX'\ee_i}^2,
\end{align*}
which means $\smallnorm{\XX'\ee_i}^2 \mathop{\approx}\limits^{\epsilon/6} \smallnorm{\tilde{\XX}\ee_i}^2$. Similarly, we can prove $\smallnorm{\YY'\ee_i}^2 \mathop{\approx}\limits^{\epsilon/6} \smallnorm{\tilde{\YY}\ee_i}^2$. Combining these relations with the initial condition, one  gets  $\smallnorm{\XX\ee_i}^2+\smallnorm{\YY\ee_i}^2 \mathop{\approx}\limits^{\epsilon/3} \smallnorm{\tilde{\XX}\ee_i}^2+\smallnorm{\tilde{\YY}\ee_i}^2$.

 For case (ii), one has
\begin{equation}
  \smallnorm{\XX'\ee_i}^2 \leq \smallnorm{\XX'\ee_i}^2+\smallnorm{\YY'\ee_i}^2 \leq 2\smallnorm{\XX'\ee_i}^2. \notag
\end{equation}
By using a similar process as above leads to
\begin{equation}
  \frac{\big|\smallnorm{\YY'\ee_i}- \smallnorm{\tilde{\YY}\ee_i}\big|}{\smallnorm{\XX'\ee_i}} \leq  2\delta n^2 \sqrt{\frac{(1+\epsilon/12)(n^2-1)}{6(1-\epsilon/12)}} \leq \frac{\epsilon}{36}.\notag
\end{equation}
Based on this obtained result, we  further have
\begin{align*}
& \big|\smallnorm{\YY'\ee_i}^2- \smallnorm{\tilde{\YY}\ee_i}^2\big|  \\
= & \big|\smallnorm{\YY'\ee_i}- \smallnorm{\tilde{\YY}\ee_i}\big|\cdot \big|\smallnorm{\YY'\ee_i}+ \smallnorm{\tilde{\YY}\ee_i}\big| \\
\leq & \frac{\epsilon}{36} \left(2+\frac{\epsilon}{36}\right) \smallnorm{\XX'\ee_i}^2
\leq  \frac{\epsilon}{12} \big(\smallnorm{\XX'\ee_i}^2+\smallnorm{\YY'\ee_i}^2\big),
\end{align*}
\begin{equation}
  \big|\smallnorm{\XX'\ee_i}^2- \smallnorm{\tilde{\XX}\ee_i}^2\big| \leq \frac{\epsilon}{12} (\smallnorm{\XX'\ee_i}^2+\smallnorm{\YY'\ee_i}^2). \notag
\end{equation}

Combining the above-obtained results, one obtains
\begin{align}
 & \frac{\big|(\smallnorm{\XX'\ee_i}^2+ \smallnorm{\YY'\ee_i}^2)-(\smallnorm{\tilde{\XX}\ee_i}^2 +\smallnorm{\tilde{\YY}\ee_i}^2)\big|}{\smallnorm{\XX'\ee_i}^2 + \smallnorm{\YY'\ee_i}^2} \notag\\
 \leq & \frac{\big|\smallnorm{\XX'\ee_i}^2-\smallnorm{\tilde{\XX}\ee_i}^2\big|+\big|\smallnorm{\YY'\ee_i}^2-\smallnorm{\tilde{\YY}\ee_i}^2\big|}{\smallnorm{\XX'\ee_i}^2 + \smallnorm{\YY'\ee_i}^2}\leq \frac{\epsilon}{6},  \notag
\end{align}
which, together with  the initial condition, leads to~\eqref{eiLLei}.
\end{proof}

%
%

\subsection{Fast algorithm for approximating $\Delta(e)$}

Based on Lemmas \ref{lem:num1}, \ref{lem:num2} and \ref{lem:dom}, we propose an algorithm $\textsc{OpinionComp}$ approximating $\Delta(e)$ for every candidate edge $e$ in  set $Q$,  the outline of  which is presented in Algorithm \ref{alg:2}, and the  performance of which is given in Theorem~\ref{lem:opcomp}.
	\begin{algorithm}
		\caption{\textsc{OpinionComp} $(\calG, Q, \epsilon, \eta)$}
		\label{alg:2}
		\Input{
			A graph $\calG$; a candidate edge set $Q$; a real number $0 < \epsilon <1/2$; a real number $1/2 < \eta < 1$ \\
		}
		\Output{
			$\{(e, \tilde{\Delta}(e)) | e \in Q\}$
			
		}
		Set $\delta_1 = \frac{\epsilon}{2n^2\sqrt{6(n^2-1)}}$,
        $\delta_2 = \frac{(1-\eta)\epsilon}{n^2\sqrt{6(n^2-1)}}$ and
        $\delta_3 =\frac{\epsilon}{72n^2} \sqrt{\frac{6(1-\epsilon/12)}{(1+\epsilon/12)(n^2-1)}}$ \;
        $t = \lceil 24\frac{\log(n-s)}{(\epsilon/12)^2}\rceil$ \;
		Compute  matrices $\bar{\BB}$ and $\WW$\ corresponding to $\LL_F$;			
		Generate random Gaussian matrices
		$\PP_{t\times \bar{m}},  \QQ_{t\times (n-s)}$\;
		Compute $\bar{\XX}=\PP\bBB$, $\bar{\YY}=\QQ \WW^{\frac{1}{2}}$ and $\bb$
		by sparse matrix multiplication in $O(tm)$ time\;
		$\hh = \SDDMSolver(\LL_F,  \boldsymbol{1},  \delta_1)$\;
		$\pp = \SDDMSolver(\LL_F,  \bb,  \delta_2)$\;
		\For{$i = 1$ to $t$}{
            $\tilde{\XX}_{[i, :]} = \SDDMSolver(\LL_{F},\bar{\XX}_{[i, :]},\delta_3)$ \;
            $\tilde{\YY}_{[i, :]} = \SDDMSolver(\LL_{F},\bar{\YY}_{[i, :]},\delta_3)$
        }
		
		\For{each $e\in Q$}{
			$j=$ the follower node to which $e$ is incident \;
			compute $\tilde{\Delta}(e) =\frac{\hh_j(1-\pp_j)}{1 + \norm{\tilde{\XX}\ee_j}^2+\norm{\tilde{\YY}\ee_j}^2}$
        }
		\Return $\{(e, \tilde{\Delta}(e)) | e \in Q\}$
	\end{algorithm}

	\begin{theorem}\label{lem:opcomp}
		For $0 < \eps < 1/2$,
		the value $\tilde{\Delta}(e)$ returned by
		$\textsc{OpinionComp}$ satisfies
		\[
		(1-\eps)\Delta(e)  \leq
		\tilde{\Delta}(e) \leq
		(1+\eps)\Delta(e)
		\]
		with high probability.
	\end{theorem}
	\begin{proof}
Using Lemmas \ref{lem:num1},~\ref{lem:num2}, and~\ref{lem:dom}, one has
			\begin{align}
			\frac{ |\Delta(e)-\tilde{\Delta}(e)| }{\Delta(e)}& \leq\frac{|(1+\eps /6)^2\Delta(e)-\Delta(e)|}{(1-\epsilon/3)\Delta(e)}\leq \eps, \nonumber
			\end{align}
as desired.
	\end{proof}

	
\subsection{Fast Algorithm for Objective Function }

By applying Algorithm~\ref{alg:2} to approximate $\Delta(e)$,  we propose a fast greedy algorithm $\textsc{Approx}(\calG, Q, k, \epsilon, \eta)$ in Algorithm~\ref{alg:3}, which solves Problem~\ref{prob:om}.  Algorithm~\ref{alg:3} performs $k$ rounds (Lines 2-6).  In every round, it takes time $\tilde{O}(m\eps^{-2})$ to  call $\textsc{OpinionComp}$ and update related qualities.  Consequently,  the time complexity of  Algorithm~\ref{alg:3} is $\tilde{O} (mk\eps^{-2})$.

\begin{algorithm}
		\caption{\textsc{Approx}$(\calG, Q,k, \epsilon,\eta)$}
		\label{alg:3}
		\Input{
			A graph $\calG$; a candidate edge set $Q$; an integer $k \leq |Q|$; a real number $0 < \epsilon < 1/2$; a real number $1/2 < \eta < 1$
		}
		\Output{
			$T$: a subset of $Q$ with $|T| = k$
		}
		Initialize solution $T = \emptyset$ \;
		\For{$i = 1$ to $k$}{
			$\{e, \tilde{\Delta}(e) | e \in Q \setminus T \} \gets \textsc{OpinionComp}(\calG, Q\setminus T, \epsilon,\eta)$ \;
			Select $e_i$ s.t.  $e_i \gets \mathrm{arg\, max}_{e \in Q \setminus T} \tilde{\Delta}(e)$ \;
			Update solution $T \gets T \cup \{ e_i \}$ \;
			Update the graph $\calG \gets \calG(V, E \cup \{ e_i \})$
		}
		\Return $T$
	\end{algorithm}

Algorithm~\ref{alg:3} yields a $\kh{1 - \frac{1}{e} - \eps}$ approximation solution to Problem~\ref{prob:om}.
\begin{theorem}\label{thm:68}
Let $T^*$ be the optimal solution to Problem \ref{prob:om}, namely,
\begin{equation*}
 H(T^*) =\argmax_{T \subset Q,|T|=k} H(T).
\end{equation*}
Then, the set $T$ returned by Algorithm \ref{alg:3} satisfies
\begin{equation*}
H(T) -H(\emptyset) \geq \left(1-\frac{1}{e}-\epsilon \right) (H(T^*) -H(\emptyset)).
\end{equation*}
\end{theorem}
\begin{proof}
The main difference between algorithms $\textsc{Exact}$ and $\textsc{Approx}$ is as follows. At each step, $\textsc{Exact}$ selects an edge with maximum marginal gain, while $\textsc{Approx}$ selects an edge with at least $(1-\epsilon)$ times the maximum marginal gain. Define $T_i=\textsc{Approx}(\calG,Q,i,\epsilon,\eta)$, then according to submodularity of $H(\cdot)$,
\begin{equation*}
H(T_{i+1}) - H(T_i) \geq \frac{1-\epsilon}{k} (H(T^*)-H(T_i))
\end{equation*}
holds for any $i$, which implies
\begin{equation*}
H(T_{i+1})-H(T^*) \leq \left(1-\frac{1-\epsilon}{k}\right) (H(T_i) -H(T^*)).
\end{equation*}
Thus, we have
\begin{align*}
H(T^*) - H(T_k)
\leq & \left(1-\frac{1-\epsilon}{k}\right)^k (H(T^*)-H(T_0)) \\
\leq & \left(\frac{1}{e}+\epsilon\right) (H(T^*)-H(T_0)).
\end{align*}
Considering $T_0 =\emptyset$, we  directly finish the proof.
\end{proof}

\section{Experiments}

In this section, we will study the performance of our two heuristic algorithms $\textsc{Exact}$ and $\textsc{Approx}$ in terms of the effectiveness and the efficiency by implementing experiments on various real-life networks with different scales. The selected data sets of real networks are publicly available in the KONECT~\cite{Ku13} and SNAP at website \url{https://snap.stanford.edu}, detailed information for whose largest components is presented in the first three columns of Table~\ref{tab:time}. For the convenience of using the linear solver  $\SDDMSolver$~\cite{KySa16}, which can be found at \url{https://github.com/danspielman/Laplacians.jl}, all our experiments are programmed in Julia using a single thread, and are run on a machine equipped with  32G RAM and 4.2 GHz Intel i7-7700 CPU. In our experiments, the candidate edge set $Q$ contains all nonexistent edges satisfying constraints in  Section~\ref{ProbStat}, for each of which, one end is in set $S_1$ and the other end is in set $F$.

	\begin{table*}[htbp]
  \centering
  \caption{The running time (seconds, $s$) and the relative error  of Algorithms~\ref{alg:1} and~\ref{alg:3} on  real-world networks for various $\epsilon$.}
  \label{tab:time}
  \fontsize{8}{8}\selectfont
  \begin{threeparttable}
    \begin{tabular}{lrrcC{1.2cm}C{1.2cm}C{1.2cm}C{1.2cm}cC{1.1cm}C{1.1cm}C{1.1cm}}
    \Xhline{2\arrayrulewidth}
    \multirow{2}{*}{Network}&
    \multirow{2}{*}{ Nodes}&
    \multirow{2}{*}{ Edges}&&
    \multicolumn{4}{c}{Running time ($s$) for $\textsc{Exact}$ and $\textsc{Approx} $ } & &
    \multicolumn{3}{c}{Relative error ($\times 10^{-2}$)} \cr
    \cmidrule{5-8} \cmidrule{10-12}
    &  & && $\textsc{Exact}$ & $\eps =0.3$ & $\eps =0.2$ & $\eps =0.1$ & &$\eps =0.3$ & $\eps =0.2$ & $\eps =0.1$ \cr
    \midrule
    494-bus             & 494   & 1,080& & 0.837  & 0.655  & 1.382  & 5.314 & &$\num{1.07}$  & $\num{0.75}$ & $\num{0.51}$\cr
    Bcspwr09            & 1,723 & 4,117 && 10.90  & 3.350  & 6.972  & 27.16 & &$\num{4.60}$  & $\num{2.43}$ & $\num{1.25}$ \cr
    Hamster         & 2,426 & 16,630& &11.27  & 3.377   & 6.936  & 29.27  &&$\num{2.21}$  & $\num{1.48}$ & $\num{0.83}$ \cr
    USGrid        		& 4,941 & 6,594 && 136.8  & 11.86   & 25.76  & 98.63 & &$\num{1.01}$  & $\num{0.48}$ & $\num{0.39}$\cr
    Bcspwr10            & 5,300 & 8,271 && 152.5  & 16.23 & 35.12  & 136.1  &&$\num{3.74}$  & $\num{1.75}$ & $\num{0.91}$\cr
    Reality	            & 6,809 & 7,680 && 300.9  & 8.715 & 18.40  & 71.56 &&$\num{0.50}$  & $\num{0.43}$ & $\num{0.17}$\cr
    PagesGovernment     & 7,057 & 89,455&& 325.4  & 93.21  & 195.8& 766.3 &&$\num{4.52}$  & $\num{2.45}$ & $\num{0.52}$\cr
    HepPh               &11,204 &117,619&& 1127   & 135.9  & 290.8  & 1120  & &$\num{5.30}$  & $\num{3.89}$ & $\num{1.95}$\cr
    Anybeat      		& 12,645 & 49,132 &&1601  & 59.43  & 128.5 & 501.1  &&$\num{0.20}$  & $\num{0.18}$ & $\num{0.08}$\cr
    Epinions            & 26,588 & 100,120 &&12515& 181.1  & 385.8 & 1568  &&$\num{0.70}$  & $\num{0.35}$ & $\num{0.09}$\cr
    EmailEU             & 32,430 & 54,397 && 22626& 79.69  & 173.1 & 703.4 &&$\num{1.49}$  & $\num{0.51}$ & $\num{0.16}$ \cr
    GemsecRO            & 41,773 &125,826 &&51236 & 378.3   & 803.2  & 3180 & &$\num{1.52}$  & $\num{1.00}$ & $\num{0.81}$\cr
    Brightkite          & 56,739 &212,945 &&--    & 443.3  & 1002   & 3827 &&--&--&--\cr
    LiveMocha           &104,103&2,193,083&& --  &  9071  & 20129  & 81203 &&--&--&--\cr
    Douban              &154,908&327,162&& --     &  1249   & 2732 & 11307 &&--&--&--\cr
    Dblp2010            &226,413&716,460&& --     &  2183  & 4837  & 19623 &&--&--&--\cr
    TwitterFollows      &404,719&713,319&& --     &  2331   & 5261  & 19942 &&--&--&--\cr
    Delicious           &536,108&1,365,961&&--   &  5540  & 12433  & 49826 &&--&--&--\cr
    FourSquare          &639,014&3,214,986&&--   & 10347  & 23214  & 92170 &&--&--&--\cr
    YoutubeSnap         &1,134,890&2,987,624&&--& 13689  & 30715 &124913 &&--&--&--\cr
    \Xhline{2\arrayrulewidth}
    \end{tabular}
    \end{threeparttable}
\end{table*}




\subsection{Effectiveness of Greedy Algorithms}

We first compare the effectiveness of our algorithms with the optimum solution and the solution of a random scheme that randomly choose $k$ edges from $Q$ to add.   To this end, we execute experiments on four small real networks: Karate club  with 34 nodes and 78 edges, Dolphins with 62 nodes and 159 edges, Tribes with 16 nodes and 58 edges, and FirmHiTech with 33 nodes and 147 edges, which allow us to compute the optimal set of added edges. We randomly select three 0-leaders and three 1-leaders. Then, by using different strategies, we add $k=1,2,\ldots,5$ edges, for each of which, one end is linked to a 1-leader  and the other end is connected to a follower. The result is reported in Figure~\ref{ComOpt1}, which shows that the solutions returned by our two greedy algorithms are the same as or very close to the optimum solution, and are far better than the  random scheme, as well as the theoretical guarantees.


\begin{figure}[tb]
		\centering
		\includegraphics[width=.95\linewidth]{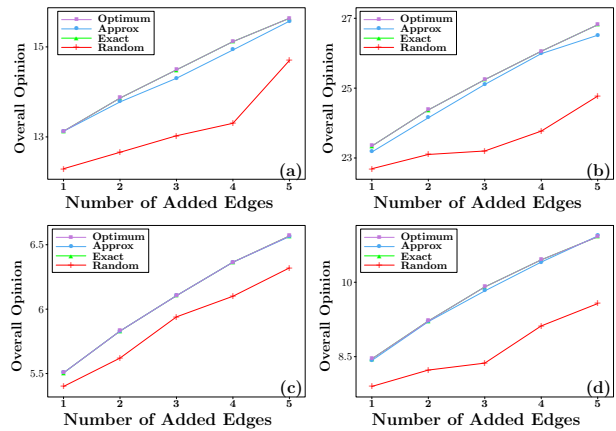}
		\caption{Overall opinion of  follower nodes  as a function of the number of added edges for our two algorithms, random and the optimum solution on four networks: Karate club  (a), Dolphins (b), Tribes (c), and FirmHiTech (d). $\epsilon$ is set to be  0.3  for the approximation algorithm $\textsc{Approx}$. \label{ComOpt1}}
\end{figure}

We also compare the results returned by our algorithms with other four baseline schemes: TopCloseness, TopBetweenness, TopPageRank and TopDegree, on relatively large real networks, in order to further demonstrate their effectiveness. The cardinalities of both  $S_0$ and  $S_1$ are equal to 10. For these centrality~\cite{MuYo19} based baselines, the added edges are just the $k$ edges linking nodes in $S_1$ and nodes in the follower set previously nonadjacent to the corresponding 1-leader, which have the highest closeness, betweenness, PageRank, and degree centrality in  original network. For each real network, we calculate the overall opinion of followers in the original graph and increase it by generating up to $k = 10,20,\ldots,50$ new edges, applying our greedy algorithms and the four baseline strategies of edge addition. After adding each edge by different methods, we compute and record the overall opinion. The results are plotted in Figure~\ref{ComBase1}, which indicates that for each network, our two greedy algorithms outperform the baseline strategies.

\begin{figure}[tb]
\centering
\includegraphics[width=.95\linewidth]{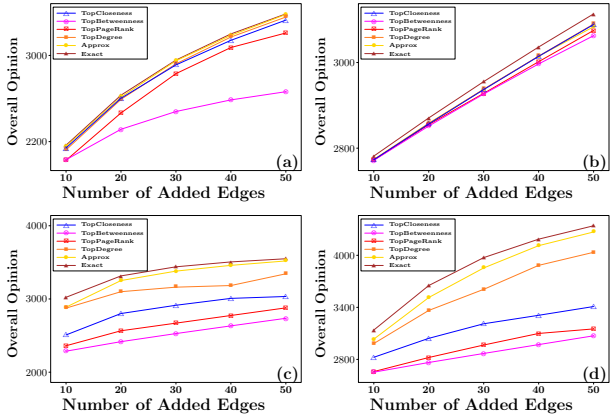}
\caption{Overall opinion of follower nodes for our two algorithms, and four baseline heuristics on four real netowrks: Reality (a), PagesGovernment (b), USgrid (c), Bcspwr10 (d). \label{ComBase1} }
\end{figure}

\subsection{Comparison of Performance between our Greedy Algorithms}

As shown above, both of our algorithms $\textsc{Exact}$ and $\textsc{Approx}$ exhibit good effectiveness, compared with baseline strategies of edge addition. Here we compare the performance of algorithms $\textsc{Exact}$ and $\textsc{Approx}$. We first demonstrate that $\textsc{Approx}$ is more efficient than $\textsc{Exact}$. For this purpose, we compare the running time of algorithms $\textsc{Exact}$ and $\textsc{Approx}$ on real-life networks in Table~\ref{tab:time}. For each network, we choose stochastically ten 1-leaders and ten 0-leaders, with the remaining nodes being followers. Then, we add $k = 50$ edges between 1-leaders and followers by algorithms $\textsc{Exact}$ and $\textsc{Approx}$ to maximize the overall opinion. Table~\ref{tab:time} shows that for moderate $\epsilon$, $\textsc{Approx}$ is much faster than $\textsc{Exact}$, which is more obvious for larger networks. In particular, $\textsc{Approx}$ is scalable to massive networks with one million nodes. For example, for the last eight networks in Table~\ref{tab:time}, such as YoutubeSnap with over $10^6$ nodes, $\textsc{Exact}$ can't run due to the memory limitation, while $\textsc{Approx}$ still works well.

We proceed to compare the effectiveness of algorithms $\textsc{Exact}$ and $\textsc{Approx}$. We define $\gamma$ and $\tilde{\gamma}$ as the increase of overall opinion of followers after adding edges selected, respectively, by  $\textsc{Exact}$ and  $\textsc{Approx}$, and define $\delta=|\gamma-\tilde{\gamma}|/\gamma$ as the relative error between $\gamma$ and $\tilde{\gamma}$. The results of relative errors for different real networks and various parameter $\epsilon$ are presented in Table~\ref{tab:time}, which demonstrates that for $\epsilon=$0.1, 0.2, and 0.3, relative errors $\delta$ are very small, with the largest value equal to $5.3\% $. Thus, the results turns by $\textsc{Approx}$ are very close to those associated with $\textsc{Approx}$, implying that $\textsc{Approx}$ is both effective and efficient.

\section{Conclusions}

In this paper, we examined the problem of maximizing the influence of opinion for leaders in a social network with $n$ nodes and $m$ edges by adding $k$ new edges based on the discrete leader-follower DeGroot model for opinion dynamics with $s \ll n$ leaders. The problem is inherently a combinatorial optimization problem that can be applied to various domains. We established the monotonicity and submodularity of the objective function. We put forward two heuristic algorithms. The former returns a $(1-\frac{1}{e})$ approximation of the optimal solution in time $O(n^3)$, while the latter has a $(1-\frac{1}{e}-\epsilon)$ approximation ratio and $\tilde{O}(km\epsilon^{-2})$ complexity. Finally, we performed experiments on real networks of different scales, demonstrating the efficiency and effectiveness of our fast algorithm that is scalable to large-scale networks with more than one million nodes. In future work, we plan to extend or modify  our algorithm to other optimization problems for opinion dynamics, such as minimizing risk of conflict, disagreement, and so on.

\section*{Acknowledgements}
The work was supported by the National Key
R \& D Program of China (Nos. 2018YFB1305104 and 2019YFB2101703), the
National Natural Science Foundation of China (Nos. 61872093,  U20B2051 and U19A2066),  the Shanghai Municipal Science and Technology
Major Project  (Nos.  2018SHZDZX01 and  2021SHZDZX0103), and ZJLab.

\bibliographystyle{ACM-Reference-Format}
\balance
\bibliography{kedges}

\end{document}